\documentclass[conference]{IEEEtran}
\usepackage{cite}

\ifCLASSINFOpdf
\else
   \usepackage[dvips]{graphicx}
   \DeclareGraphicsExtensions{.eps}
\fi
\usepackage{amsmath,amsfonts,amsthm} 
\usepackage{amssymb}
\newtheorem{thm}{Theorem}

\newtheorem{lemma}{Lemma}
\ifCLASSOPTIONcompsoc
  \usepackage[caption=false,font=normalsize,labelfont=sf,textfont=sf]{subfig}
\else
  \usepackage[caption=false,font=footnotesize]{subfig}
\fi

\newcommand{\nn}{\nonumber\\}

\begin{document}
%
\title{Status updates through M/G/1/1 queues with HARQ}

\author{\IEEEauthorblockN{Elie Najm}
\IEEEauthorblockA{LTHI, EPFL, Lausanne, Switzerland\\
Email: elie.najm@epfl.ch}
\and
\IEEEauthorblockN{Roy Yates}
\IEEEauthorblockA{ECE Dept., Rutgers University, USA\\
Email: ryates@rutgers.edu}
\and
\IEEEauthorblockN{Emina Soljanin}
\IEEEauthorblockA{ECE Dept., Rutgers University, USA\\
Email: emina.soljanin@rutgers.edu}
}


%


\maketitle

\begin{abstract}
We consider a system where randomly generated updates are to be transmitted to a monitor, but only a single update can be in the transmission service at a time. Therefore, the source has to prioritize between the two possible transmission policies: preempting the current update or discarding the new one. We consider Poisson arrivals and general service time, and refer to this system as the M/G/1/1 queue. We start by studying the average status update age and the optimal update arrival rate for these two schemes under general service time distribution. We then apply these results on two practical scenarios in which updates are sent through an erasure channel using $(a)$ an infinite incremental redundancy (IIR) HARQ system and $(b)$ a fixed redundancy (FR) HARQ system. 
We show that in both schemes the best strategy would be not to preempt. Moreover, we also prove that, from an age point of view, IIR is better than FR.   
\end{abstract}


%
\IEEEpeerreviewmaketitle

\section{Introduction}
\label{sect:sect_intro}

Previous work on status update (\cite{KaulYatesGruteser-2012Infocom,2012CISS-KaulYatesGruteser,CostaCodreanuEphremides2016,KamKompellaEphremides2013ISIT,NajmNasser-ISIT2016,YatesKaul-2012ISIT}) used an Age of Information (AoI) metric in order to assess the freshness of randomly generated updates sent by one or multiple sources to a monitor through the network. In these papers, updates are assumed to be generated according to a Poisson process and the main metric used to quantify the \emph{age} is the time average age (which we will call average age) given by
\begin{equation}
\label{eq:eq_age_definition}
\Delta = \lim_{\tau\to\infty} \frac{1}{\tau}\int_0^\tau \Delta(t)\mathrm{d}t,
\end{equation} 
where $\Delta(t)$ is the instantaneous age of the last successfully received update. If this update was generated at time $u(t)$ then its \emph{age} at time $t$ is $\Delta(t)=t-u(t)$. When the system is idle or an update is being transmitted then the instantaneous age increases linearly with time, as depicted in Fig.~\ref{fig3}. Once an update generated at time $t_i$ is received by the monitor at $t'_i$, $\Delta(t)$ drops to the value $t'_i-t_i$. This results in the sawtooth sample path seen in Fig.~\ref{fig3}. 

In this paper, we assume updates are generated according to a Poisson process with rate $\lambda$, but the system can handle only one update at a time without any buffer to store incoming updates. This means that whenever a new update is generated and the system is busy, the transmitter has to make a decision: does it give higher priority to the new update or to the one being transmitted? In other words, does it preempt or not? It has been shown that for exponential update service times, preemption ensures the lowest average age \cite{2012CISS-KaulYatesGruteser}. However, the work in \cite{NajmNasser-ISIT2016} suggests that under the assumption of gamma distributed service time, preemption might not be the best policy.

This work answers the previous question when we assume updates are sent through a symbol erasure channel with erasure rate $\delta$, while using hybrid ARQ (HARQ) protocols to combat erasures. Two HARQ protocols, introduced in \cite{heindlmaier2014isn}, are studied: $(a)$ infinite incremental redundancy (IIR)  and $(b)$  fixed redundancy (FR). In both cases we assume a generated update contains $K$ information symbols. In IIR, encoding is performed at the physical layer where the $K$ information symbols are encoded using a rateless code. Hence, the transmission of an update continues until $k_s=K$ unerased symbols are received. As for the FR, coding is applied at the physical and packet layer. This means that the update is divided into $k_p$ packets with each packet encoded using an $(n_s,k_s)$-Maximum Distance Separable (MDS) code. So, in this case, the total number of information symbols is $K=k_pk_s$. At the packet level we apply a rateless code and thus the transmission of an update terminates when $k_p$ unerased packets are received. In order to decode a packet, the receiver needs to wait for $n_s$ encoded symbols. Once received, a packet is declared erased if fewer than $k_s$ symbols are successful. It is worth noting that in this setup we send one symbol per channel use and thus the arrival rate $\lambda$ is the number of updates generated per channel use. The effect of these schemes on the transmission time of data was studied in \cite{heindlmaier2014isn}. It was shown that FR leads to a slower delivery than IIR. While the main aim of \cite{heindlmaier2014isn} is the successful delivery of every update, in this paper we are ready to sacrifice some updates for fresher information.

The impact of transmission error on the age was also investigated in \cite{ChenHuang-ISIT2016}. In this paper, service time is assumed exponential  and another age metric is used: the peak age of information. The authors conclude that, in this setup, preemption with update retransmission achieves the lowest age.

To solve the above problem, we first start by deriving in Section~\ref{sec:sec_MG11_no_preempt_general} an expression for the average age under general service time distribution when we choose not to preempt. This model is called M/G/1/1 with blocking. In Section~\ref{sec:sec_MG11_no_preempt_arq}, we use the results in the previous Section~to compute the average age when we consider the IIR and FR protocols. Sections~\ref{sec:sec_MG11_preempt_general} and \ref{sec:sec_MG11_preempt_arq} follow the same logic but  in this case we choose to preempt. This model is called M/G/1/1 with preemption. Finally, Section~\ref{sec:sec_numerical} compares the performances of both models for a given HARQ protocol as well as the performance of both protocols given a model. We show that no matter the protocol, prioritizing the current update is better than preempting it. Moreover, in the case of FR, we show that no matter the model and for a fixed arrival rate $\lambda$, there exists an optimal codeword length $n_s$.

\section{Preliminaries}
\label{sec:sec_preliminaries}
It is important to note that in both M/G/1/1 queues, some updates might be dropped. Hence we call the updates that are not dropped, and thus delivered to the receiver, as \lq\lq successfully received updates\rq\rq\ or \lq\lq successful updates\rq\rq. In addition to that, we also define: $(i)$ $I_i$ to be the true index of the $i^{th}$ successfully received update, $(ii)$ $Y_i= t'_{I_{i+1}}-t'_{I_i}$ to be the interdeparture time between two consecutive successfully received updates, $(iii)$ $X_i = t_{I_i+1} - t_{I_i}$ to be the interarrival time between the successfully transmitted update and the next generated one (which may or may not be successfully transmitted), so $f_X(x) = \lambda e^{-\lambda x}$, $(iV)$ $S_{I_{i}}$ to be the service time of the $I_{i}^{th}$ update with distribution $F_S(t)$, $(v)$ $T_i$ to be the system time, or the time spent by the $i^{th}$ successful update in the queue and $(vi)$ $N_\tau=\max\left\{n;t_{I_n}\leq\tau\right\}$, the number of successfully received updates in the interval $[0,\tau]$. In our models, we assume the service time $S_k$ of the $k^{th}$ update is independent from the interarrival time random variables $\{X_1,X_2,...,X_k,...\}$ and that the sequence $\{S_1,S_2,...\}$ forms an i.i.d process.

From \eqref{eq:eq_age_definition}, Fig.~\ref{fig3} and Fig.~\ref{fig1}, the average age for both M/G/1/1 queues can be also expressed as the sum of the geometric areas $Q_i$ under the instantaneous age curve. Authors in \cite{2012CISS-KaulYatesGruteser} show that
\begin{align}
\label{eq:eq_agedef}
\Delta	&=\lim_{\tau\to\infty}\frac{N_\tau}{\tau}\frac{1}{N_\tau}\sum_{i=1}^{N_\tau} Q_i
= \lambda_e\mathbb{E}(Q_i),
\end{align}
where $\lambda_e =  \lim_{\tau\to\infty}\frac{N_\tau}{\tau}$ and the second equality is justified by the ergodicity of the system.

\section{M/G/1/1 with Blocking}
\label{sec:sec_MG11_no_preempt_general}

In this setup, a generated update is discarded if it finds the system busy. This means an update is served only if it arrives at an idle system. This concept is depicted in Fig.~\ref{fig3}: for instance, the update generated at time $t_2$ is served since the system is empty at that time. However, the updates generated at times $t_3$ and $t_4$ find the system busy and are thus discarded. One important note here is that the system time $T_i$ of the $i^{th}$ successful update is equal to its service time.
\subsection{Average age calculation}

\begin{figure}[t]
	\centering
	\includegraphics[scale=0.5]{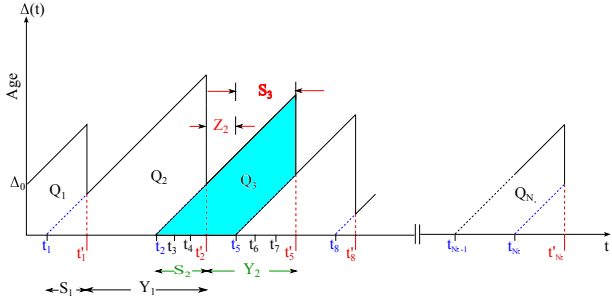}%
	\caption{Variation of the instantaneous age for M/G/1/1 with blocking}
	\label{fig3}
\end{figure}

\begin{lemma}
	\label{lemma:lemma_effective_Rate}
	For an M/G/1/1 blocking system we have, 
	\begin{equation}
	\label{eq:eq_eff_rate}
	\lambda_e = \frac{1}{\mathbb{E}(Y)}=\frac{1}{\mathbb{E}(X)+\mathbb{E}(S)},
	\end{equation}
	where $Y$, $X$ and $S$ are the steady-state counterparts of the variables defined in Section~\ref{sec:sec_preliminaries}. 
\end{lemma}
\begin{proof}
	$N_\tau$ is a renewal process with inter-renewal time between two renewals given by the random variable $Y$. As shown in Fig.~\ref{fig3}, the renewal period is the interval:
	\begin{equation}
	\label{eq:eq_renewal}
	Y_i = Z_i+S_{i+1}.
	\end{equation}
	Because each departure leaves the system empty and the interarrival times are memoryless, then the interval $Z_i$, which is the residual interarrival time until a new update is generated, is independent of $Y_{i-1}$ and it has an exponential distribution. Hence, all the $Y_i$'s are identically distributed and the $Z_i$'s are stochastically equal to the interarrival time $X$. This proves why $N_{\tau}$ is a renewal process. The claim follows \cite{ross}.
\end{proof}
Now we can compute the average age which is given by the following theorem,
\begin{thm}
	\label{thm:thm_no_preempt_age}
	The average age of an M/G/1/1 system with blocking is 
	\begin{equation}
	\label{eq:eq_no_preempt_age}
	\Delta = \mathbb{E}(S)\left(\frac{\beta}{2}(C_S+1)+\frac{1}{\beta}\right),
	\end{equation}
	where $C_S = \frac{\textrm{Var}(S)}{\mathbb{E}(S)^2}$ is the squared coefficient of variation and $\beta= \frac{\rho}{\rho+1}$ with $\rho= \frac{\mathbb{E}(S)}{\mathbb{E}(X)}=\lambda\mathbb{E}(S)$.
\end{thm}

\begin{proof}
	From (\ref{eq:eq_agedef}) we have,
	\begin{equation*}
	\Delta = \lambda_e\mathbb{E}(Q_i).
	\end{equation*}
	$\lambda_e$ is given by Lemma 1, therefore we need to compute the average area of the trapezoid $Q_i$. To do that, notice first that, using a similar argument as the one used in the proof of Lemma 1, the service time $S_i$ and $Y_i$ are independent. Thus,
	\begin{align}
	\mathbb{E}(Q_i)	&= \mathbb{E}\left(\frac{(S_{i-1}+Y_{i-1})^2}{2}-\frac{S_i^2}{2}\right)\nonumber\\
	&= \frac{1}{2}\mathbb{E}\left(Y_{i-1}^2\right) + \mathbb{E}(S_{i-1})\mathbb{E}(Y_{i-1}).
	\end{align}
	Since we are interested in the steady-state behavior, we will drop the subscript index on the random variables. Hence,
	\begin{align}
	\label{eq:eq_MG11_no_preempt_Q}
	\mathbb{E}(Q)	&= \frac{1}{2}\mathbb{E}\left(Y^2\right) + \mathbb{E}(S)\mathbb{E}(Y)\nonumber\\
	&= \frac{1}{2}\mathbb{E}\left( (X+S)^2\right)+\mathbb{E}(S)\mathbb{E}(S+X)\nonumber\\
	&= \frac{1}{2}\left(\mathbb{E}\left(X^2\right)+\mathbb{E}\left(S\right)^2\right)+\frac{1}{2}\mathrm{Var}(S)+2\mathbb{E}(S)\mathbb{E}(X)\nn
	&+\mathbb{E}(S)^2\nonumber\\
	&= \frac{1}{2}\left(\mathbb{E}\left(S\right)^2+\mathrm{Var}(S)\right)+\mathbb{E}\left(X\right)^2+2\mathbb{E}(S)\mathbb{E}(X)\nn
	&+\mathbb{E}(S)^2\nonumber\\
	&= \left(\mathbb{E}(X)+\mathbb{E}(S)\right)^2 + \frac{1}{2}\left(\mathbb{E}(S)^2+\textrm{Var}(S)\right),
	\end{align}
	where the third equality is obtained by adding and subtracting $\frac{1}{2}\mathbb{E}(S)^2$ to the second equality, and the fourth equality is obtained by noticing that for the exponential random variable $X$ we have $\mathbb{E}\left(X^2\right)=2\mathbb{E}(X)^2$.	
	Using (\ref{eq:eq_eff_rate}) and (\ref{eq:eq_MG11_no_preempt_Q}), we get (\ref{eq:eq_no_preempt_age}).
\end{proof}

\subsection{Finding the optimal arrival rate}
When the arrival rate of the updates is a parameter that we can control, it is interesting to have an idea on its value that minimizes the average age. 
\begin{thm}
	\label{thm:thm_MG11_no_preempt_opt_rate}
	For the M/G/1/1 blocking system, the minimum average age $\Delta^*$ is achieved for:
	\begin{itemize}
		\item  If $C_S>1$, then $\lambda^*=\frac{\beta^*}{(1-\beta^*)\mathbb{E}(S)}$ with $\beta^*=\sqrt{\frac{2}{C_S+1}}$ and $$\Delta^*=\mathbb{E}(S)\sqrt{2(C_S+1)}$$
		\item If $C_S\leq 1$, $\lambda^*\to\infty$  and $\Delta^*=\mathbb{E}(S)\left(\frac{1}{2}(C_S+1)+1\right)$
	\end{itemize}	
\end{thm}
\begin{proof}
	Setting the derivative of \eqref{eq:eq_no_preempt_age} with respect to $\beta$ to zero, we get:
	\begin{equation}
	\label{eq:eq_deriv_no_preempt_age}
	{\beta^*}^2 = \frac{2}{C_S+1},
	\end{equation}
	where $\beta^*$ is the optimal value of $\beta$. Since $0\leq\beta^*=\frac{\rho^*}{\rho^*+1}<1$, $C_S$ has to be strictly bigger than $1$ for $\beta^*$ to exist. In this case, $\beta^*=\sqrt{\frac{2}{C_S+1}}$ and solving for $\lambda$ we get $\lambda^*=\frac{\beta^*}{(1-\beta^*)\mathbb{E}(S)}$. Using $\beta^*$ in \eqref{eq:eq_no_preempt_age} gives the value of the minimum age $\Delta^*$.
	
	If the service time distribution is such that $C_S\leq 1$, then $\frac{\partial\Delta}{\partial\beta}=-\frac{1}{\beta^2}+\frac{C_S+1}{2}<0$. However, $\frac{\partial\beta}{\partial\lambda}= \frac{\mathbb{E}(S)}{(\lambda\mathbb{E}(S)+1)^2}\geq 0$. Therefore, $\frac{\partial\Delta}{\partial\lambda}= \frac{\partial\Delta}{\partial\beta}\frac{\partial\beta}{\partial\lambda}<0$. Thus the average age is a strictly decreasing function of the arrival rate and the minimal average age is obtained as $\lambda\to\infty$.
\end{proof}

\section{M/G/1/1 blocking HARQ system}
\label{sec:sec_MG11_no_preempt_arq}

Now, we study the effect of different HARQ policies on the average age when considering an M/G/1/1 queue without preemption. We assume that the updates are sent through a symbol erasure channel with erasure rate $\delta$. Moreover, two HARQ protocols are visited: the infinite incremental redundancy (IIR) and the fixed redundancy (FR). 
\subsection{Infinite Incremental Redundancy}
In this policy, an update consists of $k_s$ information symbols and is encoded using a rateless code. This means that the monitor needs to receive at least $k_s$ symbols in order to decode the update. The transmission of an update finishes whenever $k_s$ symbols are successfully transmitted. All updates arriving when the system is busy are discarded. Therefore, we define the  service time $S$ of an update as the number of channel uses needed for the monitor to receive $k_s$ symbols. Hence, $S$ is distributed as a negative binomial with $k_s$ successes and success probability $1-\delta$.
\begin{thm}
	\label{thm:thm_IIR_no_preemption}
	The average age of the M/G/1/1 blocking IIR-HARQ system is:
	\begin{equation}
	\label{eq:eq_IIR_age_no_preemption}
	\Delta_{\text{NIIR}}	= \frac{1}{\lambda}+\frac{k_s}{1-\delta}+\frac{\lambda k_s(k_s+\delta)}{2(1-\delta)(\lambda k_s+1-\delta)}.
	\end{equation}
	Moreover, the minimum average age is achieved for $\lambda\to\infty$ and its value is given by,
	\begin{equation}
	\label{eq:eq_min_age_IIR_no_preemption}
	\Delta_{\text{NIIR}}^* = \frac{3k_s+\delta}{2(1-\delta)}
	\end{equation}
\end{thm}

\begin{proof}
	Since we are using IIR policy then the service time $S$ of each update is distributed as a negative binomial $(k_s,1-\delta)$, $S\in\{k_s,k_s+1,\dots\}$. In this case the mean and variance  of $S$ are  given by:
	\begin{equation}
	\mathbb{E}(S)	= \frac{k_s}{1-\delta},\quad \mathrm{Var}(S)	= \frac{k_s\delta}{(1-\delta)^2}.
	\end{equation}
	Hence, we compute the quantities $\rho$, $\beta$ and $C_S$ present in (\ref{eq:eq_no_preempt_age}):
	\begin{equation}
	\rho	= \frac{\lambda k_s}{1-\delta},\quad \beta	= \frac{\rho}{\rho+1} = \frac{\lambda k_s}{\lambda k_s+1-\delta},\quad C_S	= \frac{\delta}{k_s}.
	\end{equation}
	Using the above expression in (\ref{eq:eq_no_preempt_age}) and performing some simplifications we get \eqref{eq:eq_IIR_age_no_preemption}.
	
	Moreover, since $\delta\leq 1$ and $k_s\geq 1$, $C_S = \frac{\delta}{k_s}\leq 1$. By Theorem~\ref{thm:thm_MG11_no_preempt_opt_rate}, the optimum average age is achieved as $\lambda\to\infty$. Taking the limit on \eqref{eq:eq_IIR_age_no_preemption} gives \eqref{eq:eq_min_age_IIR_no_preemption}.
\end{proof}

\subsection{Fixed Redundancy}
In this policy, we apply two levels of coding: a packet level and a physical level. Each update consists of $k_p$ packets encoded using a rateless code. This means that the monitor needs to receive $k_p$ decodable packets in order to decode the update. Moreover, each packet contains $k_s$ information symbols and is encoded using a $(n_s,k_s)$-Maximum Distance Separable (MDS) code. Hence, a packet can be decoded if at least $k_s$ symbols are not erased. Since the packets are being transmitted through a symbol erasure channel with erasure probability $\delta$ than the probability for the receiver to be unable to decode a packet is:
\begin{align}
\label{eq:eq_ep}
\epsilon_p	&= \mathbb{P}(\text{less than $k_s$ symbols received})\nn
&= \sum_{i=0}^{k_s-1}{{n_s}\choose{i}}\delta^{n_s-i}(1-\delta)^i.
\end{align}
\begin{thm}
	\label{thm:thm_FR_no_preemption}
	The average age of the M/G/1/1 FR-HARQ blocking system is
	\begin{equation}
	\label{eq:eq_FR_age_no_preemption}
	\Delta_{\text{NFR}}	= \frac{1}{\lambda}+\frac{n_sk_p}{1-\epsilon_p}+\frac{\lambda n_s^2k_p(k_p+\epsilon_p)}{2(1-\epsilon_p)(\lambda n_sk_p+1-\epsilon_p)}.
	\end{equation}
	
	Moreover, the minimum average age is achieved as $\lambda\to\infty$ and its value is given by,
	\begin{equation}
	\label{eq:eq_min_age_FR_no_preemption}
	\Delta_{\text{NFR}}^* = \frac{3n_sk_p+\epsilon_p}{2(1-\epsilon_p)}
	\end{equation}
\end{thm}

\begin{proof}
	The number $M$ of packets needed to be transmitted to decode an update is distributed as a negative binomial $(k_p, 1-\epsilon_p)$ random variable with $k_p$ successes and success rate $(1-\epsilon_p)$, $M\in\{k_p,k_p+1,\dots\}$. Since the transmission of each packet consumes $n_s$ channel uses then the service time $S$ of each update is $S=n_sM$. Thus, the mean and variance of $S$ are given by:
	\begin{align}
	\mathbb{E}(S)	&= \mathbb{E}(n_sM) = n_s\mathbb{E}(M)=\frac{n_sk_p}{1-\epsilon_p},\\ 
	\mathrm{Var}(S) &= \mathrm{Var}(n_sM)=n_s^2\mathrm{Var}(M)=\frac{n_s^2k_p\epsilon_p}{(1-\epsilon_p)^2}.
	\end{align}
	Hence, we compute the quantities:
	\begin{equation}
	\rho	= \frac{\lambda k_p}{1-\epsilon_p},\quad \beta	= \frac{\lambda k_p}{\lambda k_p+1-\epsilon_p},\quad C_S	= \frac{\epsilon_p}{k_p}.
	\end{equation}
	Using the above expressions in (\ref{eq:eq_no_preempt_age}) and performing some simplifications we get \eqref{eq:eq_FR_age_no_preemption}.
	
	Moreover, since $\epsilon_p\leq 1$ and $k_p\geq 1$, $C_S = \frac{\epsilon_p}{k_p}\leq 1$. By Theorem~\ref{thm:thm_MG11_no_preempt_opt_rate}, the optimum average age is achieved as $\lambda\to\infty$. From \eqref{eq:eq_FR_age_no_preemption} this yields \eqref{eq:eq_min_age_FR_no_preemption}.
\end{proof}

\section{M/G/1/1 with preemption}
\label{sec:sec_MG11_preempt_general}

In the M/G/1/1 with preemption scenario, any packet being served is preempted if a new packet arrives and the new packet is served instead. In fact, while in the M/G/1/1 with blocking the priority is given to the update being served, in this setup the priority goes to the newly generated update. Moreover, the number of packets in the queue can be modeled as a continuous-time two-state semi-Markov chain depicted in Figure~\ref{fig2}.

The 0-state corresponds to empty queue and no packet is being served while the 1-state corresponds to the state where the queue is full and is serving one packet. However, given that the interarrival time between packets is exponentially distributed with rate $\lambda$ then one spends an exponential amount of time $X$ in the 0-state before jumping with probability 1 to the other state. Once in the 1-state, two independent clocks are started: the service time clock of the packet being served and the rate $\lambda$ memoryless clock of the interarrival time between the current packet and the next one to be generated. If the memoryless clock ticks first, we stay in the 1-state, otherwise we go back to the 0-state. Hence, the jump from the 1-state to the 0-state occurs with probability $p = \mathbb{P}(S<X)$, where $S$ is a generic service time with distribution $f_S(t)$ and $X$ is a generic rate $\lambda$ memoryless interarrival time which is independent of $S$.

The quantity $p$ will play an important role in our derivation, so we will take a closer look at it:
\begin{align}
\label{eq1}
p &= \int_0^\infty\!f_S(t)\mathbb{P}(X>t)\ \mathrm{d}t
= \int_0^\infty\!f_S(t)e^{-\lambda t}\ \mathrm{d}t
= P_\lambda,
\end{align}
where $P_\lambda$ is the Laplace transform of the service time distribution.

\begin{figure}[!t]
	\centering
	\includegraphics[scale=0.3]{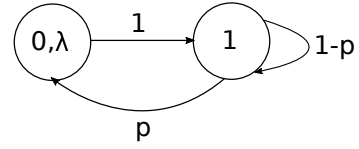}
	\caption{Semi-Markov chain representing the queue for LCFS with preemption}
	\label{fig2}
\end{figure}

Using Fig.~\ref{fig1} it was shown in \cite{NajmNasser-ISIT2016} that the average age $\Delta$ is:
\begin{align}
\label{eq2}
\Delta = \lambda_e\mathbb{E}(Q)=\lambda_e\left(\frac{1}{2}\mathbb{E}\left(Y^2\right)+\mathbb{E}(T)\mathbb{E}(Y)\right),
\end{align}
where $\lambda_e=\lambda P_\lambda$ is the effective arrival rate, $T$ and $Y$ as defined in Section~\ref{sec:sec_preliminaries}. We start with $\mathbb{E}(T)$.
\begin{figure}[!t]
	\centering
	\includegraphics[scale=0.5]{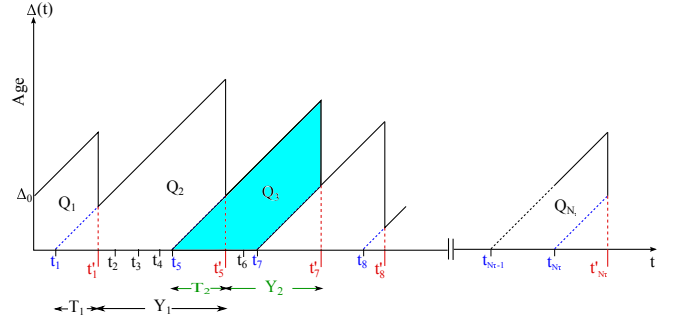}%
	\caption{Variation of the instantaneous age for LCFS with preemption}
	\label{fig1}
\end{figure}

\begin{lemma}
	\label{lemma:lemma_MG11_preempt_f(T)}
	The PDF of the system time $T$ is 
	\begin{align}
	\label{eq3}
	f_T(t)	&= \frac{f_S(t)}{P_\lambda}e^{-\lambda t}.
	\end{align}
	Its expected value is
	\begin{equation}
	\label{eq:eq_MG11_preempt_E(T)}
	\mathbb{E}(T)	= -\frac{1}{P_\lambda}\frac{\partial P_\lambda}{\partial \lambda}.
	\end{equation}
\end{lemma}
\begin{proof}
	\begin{align}
	\label{eq:eq_f(T)}
	f_T(t)	
	&= \lim_{\epsilon\to 0} \frac{\mathbb{P}(S\in [t,t+\epsilon]| S<X)}{\epsilon}\nonumber\\
	&= \lim_{\epsilon\to 0} \frac{\mathbb{P}(S\in[t,t+\epsilon])}{\epsilon P_\lambda}\mathbb{P}(S<X|S\in(t,t+\epsilon))\nonumber\\
	&= \frac{f_S(t)}{P_\lambda}\mathbb{P}(X>t) =\frac{f_S(t)}{P_\lambda}e^{-\lambda t}.
	\end{align}
	Using (\ref{eq3}) we calculate the expected value of $T$:
	\begin{align}
	\label{eq4}
	\mathbb{E}(T)	
	&= \frac{1}{P_\lambda}\int_0^\infty tf_S(t)e^{-\lambda t}\mathrm{d}t= -\frac{1}{P_\lambda}\frac{\partial P_\lambda}{\partial \lambda}.
	\end{align}
\end{proof}
%
Now we only need to calculate the first and second moments of $Y$. For that we will derive its moment generating function.
\begin{lemma}
	\label{lemma:lemma_MG11_preempt_MGF(Y)}
	The moment generating function of the interdeparture time $Y$ is given by
	\begin{align}
	\label{eq:eq_MG11_preempt_MGF(Y)}
	\phi_Y(s)	&= \frac{\lambda P_{\lambda-s}}{\lambda P_{\lambda-s} -s},
	\end{align}
	where $P_{\lambda-s} = \int_0^\infty f_S(t)e^{-(\lambda -s)}\ \mathrm{d}t$.
\end{lemma}

\begin{proof}
	From Fig.~\ref{fig1} we can deduce that $Y$ is the shortest time to go from the 0-state back to the 0-state. This means that
	\begin{equation}
	\label{eq5}
	Y = X + W,
	\end{equation}
	where  $X$ is exponentially distributed with rate $\lambda$ and $W$ is 
	\begin{align}
	\label{eq6}
	W &= \left\{\begin{array}{ll}
	T & \text{with probability\ } p\\
	X'_1+T & \text{with probability\ } (1-p)p\\
	X'_1+X'_2+T & \text{with probability\ } (1-p)^2p\\
	\vdots
	\end{array}\right.\nonumber\\
	&= \sum_{j=0}^M X'_j+T,
	\end{align}
	where $X'_0=0$ and for $j>0$, $X'_j$ is such that $\mathbb{P}(X'_j<\alpha)=\mathbb{P}(X<\alpha|X<S)$. $M$, which gives the number of discarded packets before the first successful reception, is a geometric$(p)$ random variable independent of $X'_j$ and $T$.
	We start first by deriving the moment generating function of $X'$.
	\begin{align}
	f_{X'}(t)	
	&= \lim_{\epsilon\to 0} \frac{\mathbb{P}(X\in [t,t+\epsilon]| S>X)}{\epsilon}\nn
	&= \lim_{\epsilon\to 0} \frac{\mathbb{P}{X\in[t,t+\epsilon])}}{\epsilon (1-P_\lambda)}\mathbb{P}(S>X|X\in(t,t+\epsilon))\nn
	&= \frac{f_X(t)}{1-P_\lambda}\mathbb{P}(S>t)\nn
	f_{X'}(t)	&= \left[1-F_S(t)\right]\frac{\lambda e^{-\lambda t}}{1-P_\lambda},
	\end{align}
	where $F_S(t)$ is the cdf of the service time $S$.
	Hence,
	\begin{align}
	\label{eq7}
	\phi_{X'}(s)	&= \mathbb{E}\left(e^{sX'}\right) = \int_0^\infty e^{st} \left(1-F_S(t)\right)\frac{\lambda e^{-\lambda t}}{1-P_\lambda}\ \mathrm{d}t\nonumber\\
	&\overset{(a)}{=} \frac{\lambda}{\lambda-s}\frac{1}{1-P_\lambda} - \frac{\lambda}{1-P_\lambda}\frac{P_{\lambda-s}}{\lambda-s}\nonumber\\
	&= \frac{\lambda(1-P_{\lambda-s})}{(\lambda-s)(1-P_\lambda)},
	\end{align}
	where $(a)$ is obtained by using integration by parts with $u=1-F_S(t)$ and $\frac{\mathrm{d}v}{\mathrm{d}t}=e^{-t(\lambda-s)}$.
	On the other hand, \eqref{eq3} implies
	\begin{align}
	\label{eq8}
	\phi_T(s)	&= \mathbb{E}\left(e^{sT}\right)
	= \int_0^\infty \frac{f_S(t)}{P_\lambda}e^{-\lambda t}e^{st}\ \mathrm{d}t
	= \frac{P_{\lambda-s}}{P_\lambda}.
	\end{align}
	Using (\ref{eq7}) and (\ref{eq8}), we deduce the moment generating of $W$,
	\begin{align}
	\label{eq9}
	\phi_W(s)	&= \mathbb{E}\left(e^{s\left(\sum_{i=0}^M X'_i  + T \right)}\right)\nonumber\\
	&= \mathbb{E}\left(e^{sT}\right)\mathbb{E}\left(\mathbb{E}\left(e^{sX'}\right)^M\right)\nonumber\\
	&= \frac{P_{\lambda-s}}{P_\lambda} \sum_{i=0}^\infty \left(\frac{\lambda(1-P_{\lambda-s})}{(\lambda-s)(1-P_\lambda)}\right)^i(1-P_\lambda)^iP_\lambda\nonumber\\
	&= \frac{(\lambda-s)P_{\lambda-s}}{\lambda P_{\lambda-s}-s}.
	\end{align}
	Using (\ref{eq9}) and that $\phi_X=\mathbb{E}\left(e^{sX}\right)=\frac{\lambda}{\lambda-s}$, we get \eqref{eq:eq_MG11_preempt_MGF(Y)} from $\phi_Y(s)=\mathbb{E}\left(e^{sX}\right)\mathbb{E}\left(e^{sW}\right)$.
\end{proof}
\begin{thm}
	\label{thm:thm_MG11_preempt_general_age}
	The average age of an M/G/1/1 system with preemption is given by,
	\begin{equation}
	\label{eq13}
	\Delta	= \lambda_e\mathbb{E}(Q) = \frac{1}{\lambda P_\lambda}.
	\end{equation}
\end{thm}
\begin{proof}
	Deriving \eqref{eq:eq_MG11_preempt_MGF(Y)} once and twice and setting $s=0$ gives:
	\begin{equation}
	\label{eq11}
	\mathbb{E}(Y)	= \frac{1}{\lambda P_\lambda}  \quad \text{and}\quad
	\mathbb{E}(Y^2)	= \frac{2}{\lambda^2 P_\lambda^2}\left(1+\lambda \frac{\partial P_\lambda}{\partial\lambda}\right)
	\end{equation}
	Using \eqref{eq:eq_MG11_preempt_E(T)} and (\ref{eq11}) we get $\mathbb{E}(Q) 	= \frac{1}{\lambda^2P_\lambda^2}$.
	This last expression and the fact that $\lambda_e=\lambda P_\lambda$ give \eqref{eq13}.
\end{proof}

In conclusion, for the M/G/1/1 with preemption, the average age depends on the Laplace transform of the service time distribution.

\section{M/G/1/1 with preemption and HARQ}
\label{sec:sec_MG11_preempt_arq}

In this Section~we study the effect of different HARQ policies on the average age when considering an M/G/1/1 queue with preemption. Indeed, we assume that the updates are sent through a symbol erasure channel with erasure rate $\delta$. Moreover, two HARQ models are visited: the infinite incremental redundancy (IIR) and the fixed redundancy (FR). 
\subsection{Infinite Incremental Redundancy}
In this setup, the transmission of an update finishes whenever one of these events happen first: $(i)$ $k_s$ symbols are successfully transmitted, or $(ii)$ a new update is generated. Hence the following theorem.
\begin{thm}
	\label{thm:thm_IIR_preemption}
	The average age of an M/G/1/1 with preemption system when using the IIR policy is given by,
	\begin{equation}
	\label{eq:eq_IIR_age_preemption}
	\Delta_{\text{PIIR}}	= \frac{1}{\lambda}\left(\frac{e^{\lambda}-\delta}{1-\delta}\right)^{k_s}.
	\end{equation}
	Moreover, $\Delta_{\text{PIIR}}$ has a minimum and the arrival rate $\lambda^*$ that achieves it should satisfy the condition
	\begin{equation}
	\label{eq:eq_MG11_with_preempt_IIR_cond}
	\lambda^*\leq \frac{1}{k_s}.
	\end{equation}
	The minimum age $\Delta_{\text{PIIR}}^*$ can be lower bounded using
	\begin{equation}
	\label{eq:eq_MG11_with_preemption_IIR_age_min}
	\Delta_{\text{PIIR}}^* \geq \frac{1}{\lambda_{\text{IIR}}}\left(1+\frac{\lambda_{\text{IIR}}}{1-\delta}\right)^{k_s},
	\end{equation}
	where $\lambda^*\approx\lambda_{\text{IIR}}= \frac{1-k_s+\sqrt{(k_s+1)^2-4k_s\delta}}{2k_s}$.
\end{thm}

\begin{proof}
	Under the IIR policy, the service time $S$ of each update is distributed as a negative binomial $(k_s,1-\delta)$, $S\in\{k_s,k_s+1,\dots\}$. In this case the moment generating function of $S$ is given by:
	\begin{equation}
	\label{eq:eq_mgf_NB}
	\phi_S(s) = \mathbb{E}\left(e^{sS}\right) = \left(\frac{1-e^{s}\delta}{e^{s}(1-\delta)}\right)^{-k_s}.
	\end{equation}
	Noting that $P_\lambda=\phi_S(-\lambda)$ and using (\ref{eq13}) and (\ref{eq:eq_mgf_NB}), we get (\ref{eq:eq_IIR_age_preemption}).
	To prove condition \eqref{eq:eq_MG11_with_preempt_IIR_cond} we differentiate $\Delta_{\text{PIIR}}$ with respect to $\lambda$ and equate it to zero. This yields
	\begin{align}
	\label{eq:eq_MG11_age_deriv}
	&-\frac{1}{\lambda}\left(\frac{e^\lambda-\delta}{1-\delta}\right)+\frac{k_se^\lambda}{1-\delta}	=0.
	\end{align}
	Thus, to satisfy \eqref{eq:eq_MG11_age_deriv} we need 
	\begin{align}
	\label{eq:eq_MG11_age_deriv_reduction}
	e^\lambda(k_s\lambda-1)						&=-\delta.
	\end{align} 
	Since $0\leq\delta\leq 1$, \eqref{eq:eq_MG11_age_deriv_reduction} implies that $k_s\lambda-1\leq0$. Hence \eqref{eq:eq_MG11_with_preempt_IIR_cond} holds. 
	Moreover, since $\lambda>0$, we have that $e^\lambda>1+\lambda$. This means that if $\lambda^*$ minimizes $\Delta_{\text{PIIR}}$, then
	\begin{align}
	\label{eq:eq_MG11_preempt_lower_bound}
	\Delta_{\text{PIIR}}^*=\Delta_{\text{PIIR}}(\lambda^*)	
	&> \frac{1}{\lambda^*}\left(1+\frac{\lambda^*}{1-\delta}\right)^{k_s}.
	\end{align} 
	Finally, in order to obtain $\lambda^*$ one needs to solve equation \eqref{eq:eq_MG11_age_deriv_reduction} which does not have a simple closed form expression. As an alternative, we can make the small $\lambda$ approximation $e^{\lambda^*}\approx 1+\lambda^*$. In this case, \eqref{eq:eq_MG11_age_deriv_reduction} reduces to
	\begin{equation}
	(1+\lambda)(k_s\lambda-1)=-\delta.
	\end{equation}
	This is a quadratic equation whose only positive root is given by $$\lambda_{IIR}= \frac{1-k_s+\sqrt{(k_s+1)^2-4k_s\delta}}{2k_s}.$$
	To obtain \eqref{eq:eq_MG11_with_preemption_IIR_age_min}, we replace $\lambda^*$ by $\lambda_{\text{IIR}}$ in \eqref{eq:eq_MG11_preempt_lower_bound}.
\end{proof}
Since $\lambda^*\leq \frac{1}{k_s}\leq 1$,  the lower bound in \eqref{eq:eq_MG11_with_preemption_IIR_age_min} becomes a tight approximation of the average age for typical values of $k_s$.
\subsection{Fixed Redundancy}
In this case also the transmission of an update is terminated whenever one of these events happen first: $(i)$ $k_p$ packets are successfully transmitted, or $(ii)$ a new update is generated. As in the M/G/1/1 blocking system, we define the packet erasure probability $\epsilon_p=\sum_{i=0}^{k_s-1}{{n_s}\choose{i}}\delta^{n_s-i}(1-\delta)^i$.
\begin{thm}
	\label{thm:thm_FR_preemption}
	The average age of the information for an M/G/1/1 with preemption system using the FR policy is given by,
	\begin{equation}
	\label{eq:eq_FR_age_preemption}
	\Delta_{\text{PFR}}	=  \frac{1}{\lambda}\left(\frac{1-e^{-\lambda n_s}\epsilon_p}{e^{-\lambda n_s}(1-\epsilon_p)}\right)^{k_p}.
	\end{equation}
	Moreover, $\Delta_{\text{PFR}}$ has a minimum and the arrival rate $\lambda^*$ that achieves it should satisfy the condition
	\begin{equation}
	\label{eq:eq_MG11_with_preempt_FR_cond}
	\lambda^*\leq \frac{1}{n_sk_p}.
	\end{equation}
	The minimum age $\Delta_{\text{PIIR}}^*$ can be lower bounded using
	\begin{equation}
	\label{eq:eq_MG11_with_preemption_FR_age_min}
	\Delta_{\text{PFR}}^* \geq \frac{1}{\lambda_{\text{FR}}}\left(1+\frac{\lambda_{\text{FR}}n_s}{1-\epsilon_p}\right)^{k_p},
	\end{equation}
	where $\lambda^*\approx \lambda_{\text{FR}}= \frac{1-k_p+\sqrt{(k_p+1)^2-4k_p\epsilon_p}}{2n_sk_p}$.
\end{thm}
\begin{proof}
	The number $M$ of packets needed to be transmitted to decode an update is distributed as a negative binomial $(k_p, 1-\epsilon_p)$ random variable with $k_p$ successes and success rate $(1-\epsilon_p)$, $M\in\{k_p,k_p+1,\dots\}$. Since the transmission of each packet consumes $n_s$ channel uses, the service time $S$ of each update is $S=n_sM$. Thus, the moment generating function of $S$ is:
	\begin{equation}
	\phi_S(s)	=\mathbb{E}\left(e^{sn_sM}\right) = \phi_M(n_s s) = \left(\frac{e^{n_s s}(1-\epsilon_p)}{1-e^{n_s s}\epsilon_p}\right)^{k_p}.
	\end{equation}
	Using (\ref{eq13}), the fact that $P_\lambda=\phi_S(-\lambda)$ and the above expression we obtain (\ref{eq:eq_FR_age_preemption}).
	
	To prove condition \eqref{eq:eq_MG11_with_preempt_FR_cond} we differentiate $\Delta_{PFR}$ with respect to $\lambda$ and equate it to zero, yielding
	\begin{align}
	\label{eq:eq_MG11_FR_age_deriv}
	&-\frac{1}{\lambda}\left(\frac{e^{\lambda n_s}-\epsilon_p}{1-\epsilon_p}\right)+\frac{k_pn_se^{\lambda n_s}}{1-\epsilon_p}	=0.
	\end{align}
	Thus, to satisfy \eqref{eq:eq_MG11_FR_age_deriv} we need 
	\begin{align}
	\label{eq:eq_MG11_FR_age_deriv_reduction}
	e^{\lambda n_s}(k_pn_s\lambda-1)						&=-\epsilon_p.
	\end{align} 
	Since $0\leq\epsilon_p\leq 1$, \eqref{eq:eq_MG11_FR_age_deriv_reduction} implies that $k_pn_s\lambda-1\leq0$. Hence \eqref{eq:eq_MG11_with_preempt_FR_cond} holds.
	
	As in the proof for Theorem~\ref{thm:thm_IIR_preemption}, here also we have:
	\begin{align}
	\label{eq:eq_MG11_preempt_FR_lower_bound}
	\Delta_{\text{PFR}}^*=\Delta_{\text{PFR}}(\lambda^*)	&> \frac{1}{\lambda^*}\left(1+\frac{n_s\lambda^*}{1-\epsilon_p}\right)^{k_p}.
	\end{align} 
	
	Finally, also as in the proof for Theorem~\ref{thm:thm_IIR_preemption}, we approximate the real value of $\lambda^*$ by solving the quadratic equation
	\begin{equation}
	(1+\lambda n_s)(k_pn_s\lambda-1)=-\epsilon_p.
	\end{equation}
	The only positive root is given by 
	$$ \lambda_{\text{FR}}= \frac{1-k_p+\sqrt{(k_p+1)^2-4k_p\epsilon_p}}{2n_sk_p}.$$
	To obtain \eqref{eq:eq_MG11_with_preemption_FR_age_min}, we replace $\lambda^*$ by $\lambda_{FR}$ in \eqref{eq:eq_MG11_preempt_FR_lower_bound}.
\end{proof}
Since $n_s\lambda^*\leq \frac{1}{k_p}\leq 1$,  the lower bound in \eqref{eq:eq_MG11_with_preemption_FR_age_min} becomes a tight approximation for typical values of $k_p$.

\section{Numerical results}
\label{sec:sec_numerical}

In this section, we first compare the two HARQ  policies, IIR and FR, for the M/G/1/1 with and without preemption. Then, for each HARQ policy, we compare the performances of the two M/G/1/1 schemes. Moreover, for the simulation results discussed in this section, we assume the following setting: a symbol erasure channel with erasure rate $\delta=0.2$ and each update in IIR-HARQ and FR-HARQ contain $K=100$ information symbols. So for IIR-HARQ we have $f_s=100$ while for FR-HARQ, we assume each update is divided into $k_p=K/k_s$ packets where each packet is encoded using an MDS-$(k_s,n_s)$ code.   

We first start analyzing the M/G/1/1 system with preemption. Fig.~\ref{fig:fig_MG11_with_preempt_diff_ks} shows the average age for different values of $k_s$ around its minimum point. As we can notice, if we choose the optimum $n_s$ for a fixed $k_s$ and range of $\lambda$ then the average age decreases as the number of packets per update decreases. In fact, the black curve which corresponds to $k_p=1$ has the lowest average age around its minimum, followed by the blue curve associated with $k_p=5$ and the worst performance is for the system with $k_p=10$.  Fig.~\ref{fig:fig_MG11_with_preempt_diff_ks} also confirms the results in Theorem~\ref{thm:thm_IIR_preemption} and Theorem~\ref{thm:thm_FR_preemption} saying that $\Delta_{\text{PIIR}}$ and $\Delta_{\text{PFR}}$ achieve a minimum at a small value of $\lambda$. This figure also suggests that no matter how we choose $k_s$ and $n_s$, IIR outperforms FR. The values of $n_s$ chosen in Fig.~\ref{fig:fig_MG11_with_preempt_diff_ks} are such that they minimize the average age for a given $\delta$ and $k_s$. The existence of such optimum packet length in FR can be deduced from Fig.~\ref{fig:fig_MG11_preempt_ns}.  Here we set $\lambda=0.0066$, which minimizes the average age for $\delta=0.2$, and $k_s=20$. Fig.~\ref{fig:fig_MG11_preempt_ns} can be explained using the lower bound \eqref{eq:eq_MG11_with_preemption_FR_age_min}: for a given $\lambda$, as $n_s$ gets large, $\epsilon_p\to 0$ and the lower bound will be increasing with $n_s$ since $\left(1+\frac{n_s\lambda^*}{1-\epsilon_p}\right)>1$. However, for $n_s$ close to $k_s$, $\epsilon_p\to 1$ which also increases this lower bound. Thus, the packet length should be neither too small (equal to $k_s$) nor too large. As it is expected, Fig.~\ref{fig:fig_MG11_preempt_ns} also shows that the optimal packet length $n_s$ increases as the erasure rate $\delta$ increases.

The above results concerning the M/G/1/1 system with preemption apply also for the M/G/1/1 blocking system as it can be seen in Fig.~\ref{fig:fig_MG11_no_preempt_diff_ks} and \ref{fig:fig_MG11_no_preempt_ns}. However, some differences need to be noted. $(i)$ Fig.~\ref{fig:fig_MG11_no_preempt_diff_ks} confirms the results of Theorems~\ref{thm:thm_IIR_no_preemption} and \ref{thm:thm_FR_no_preemption} that the average age is a decreasing function of $\lambda$. $(ii)$ Fig.~\ref{fig:fig_MG11_no_preempt_diff_ks} shows that for any value of $\lambda$, increasing the number of packets per update increases the average age. $(iii)$ Fig.~\ref{fig:fig_MG11_no_preempt_ns} shows the existence of an optimal packet length $n_s$ for a given $\delta$, $\lambda$ and $k_s$.

Finally, we compare the performance of the M/G/1/1 with preemption and the M/G/1/1 blocking systems for each one of the HARQ policies. In both cases, Fig.~\ref{fig:fig_MG11_all_schemes_IIR_FR} shows that the M/G/1/1 blocking system performs better than its counterpart for all values of $\lambda$.

\begin{figure}[!t]
	\centering
	\includegraphics[scale=0.4]{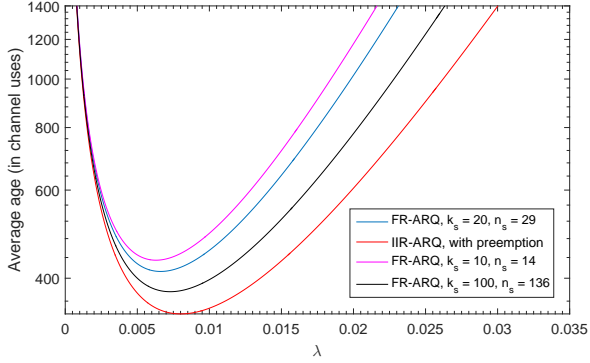}
	\caption{Comparing the performance of the FR-HARQ for the M/G/1/1 with preemption scheme when varying the number of information symbols in each packet. We assume the update has $100$ information symbols, $\delta=0.2$, $k_p=100/k_s$. $n_s$ is chosen to minimize the average age.}
	\label{fig:fig_MG11_with_preempt_diff_ks}
\end{figure}


\begin{figure}[!t]
	\centering
	\includegraphics[scale=0.4]{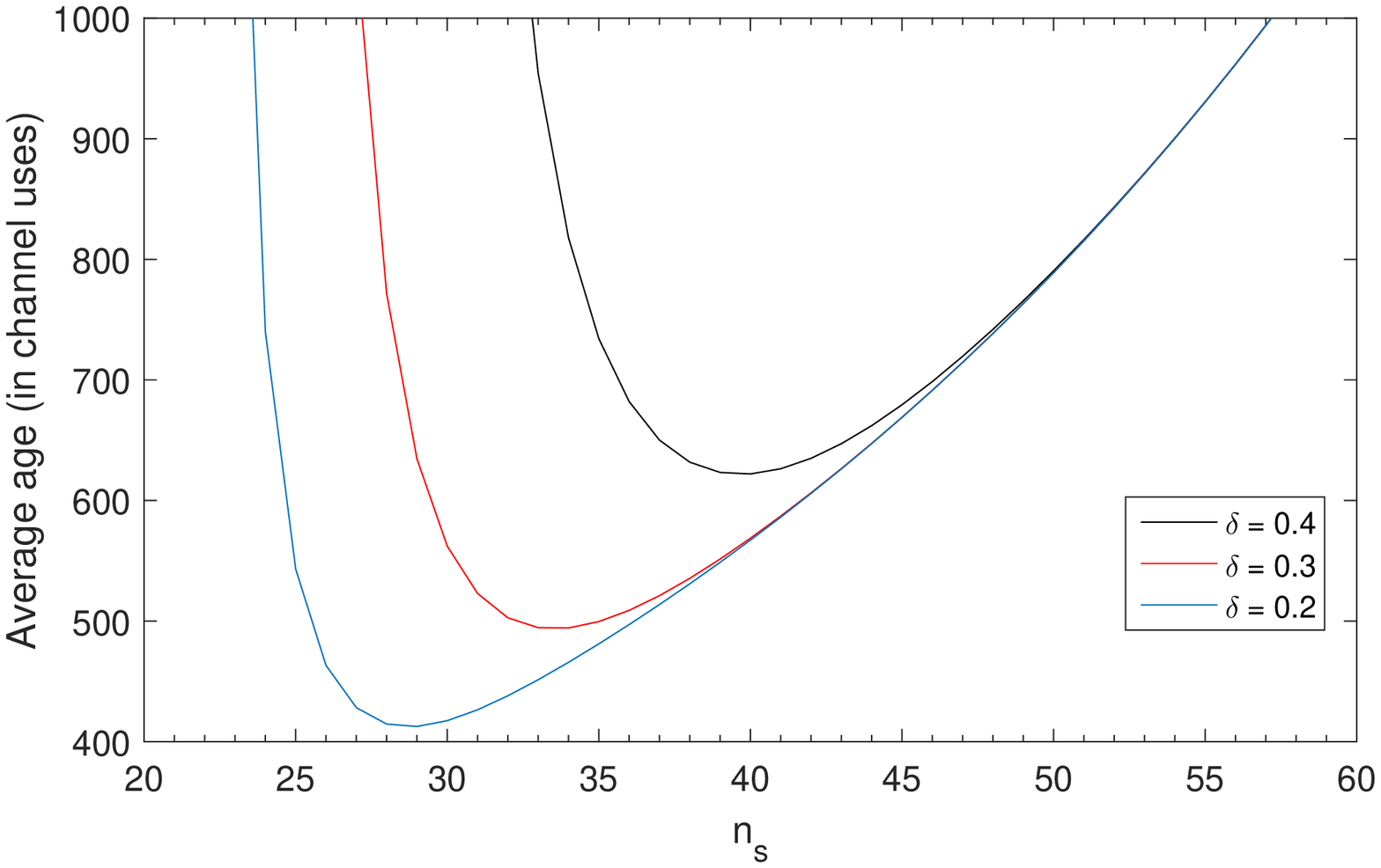}
	\caption{Average age with respect to codeword length for the M/G/1/1 with preemption scheme with FR-HARQ. We assume the update has $100$ information symbols, $\lambda=0.0066$, $k_s=20$ and $k_p=100/k_s$.}
	\label{fig:fig_MG11_preempt_ns}
\end{figure}


\begin{figure}[!t]
	\centering
	\includegraphics[scale=0.4]{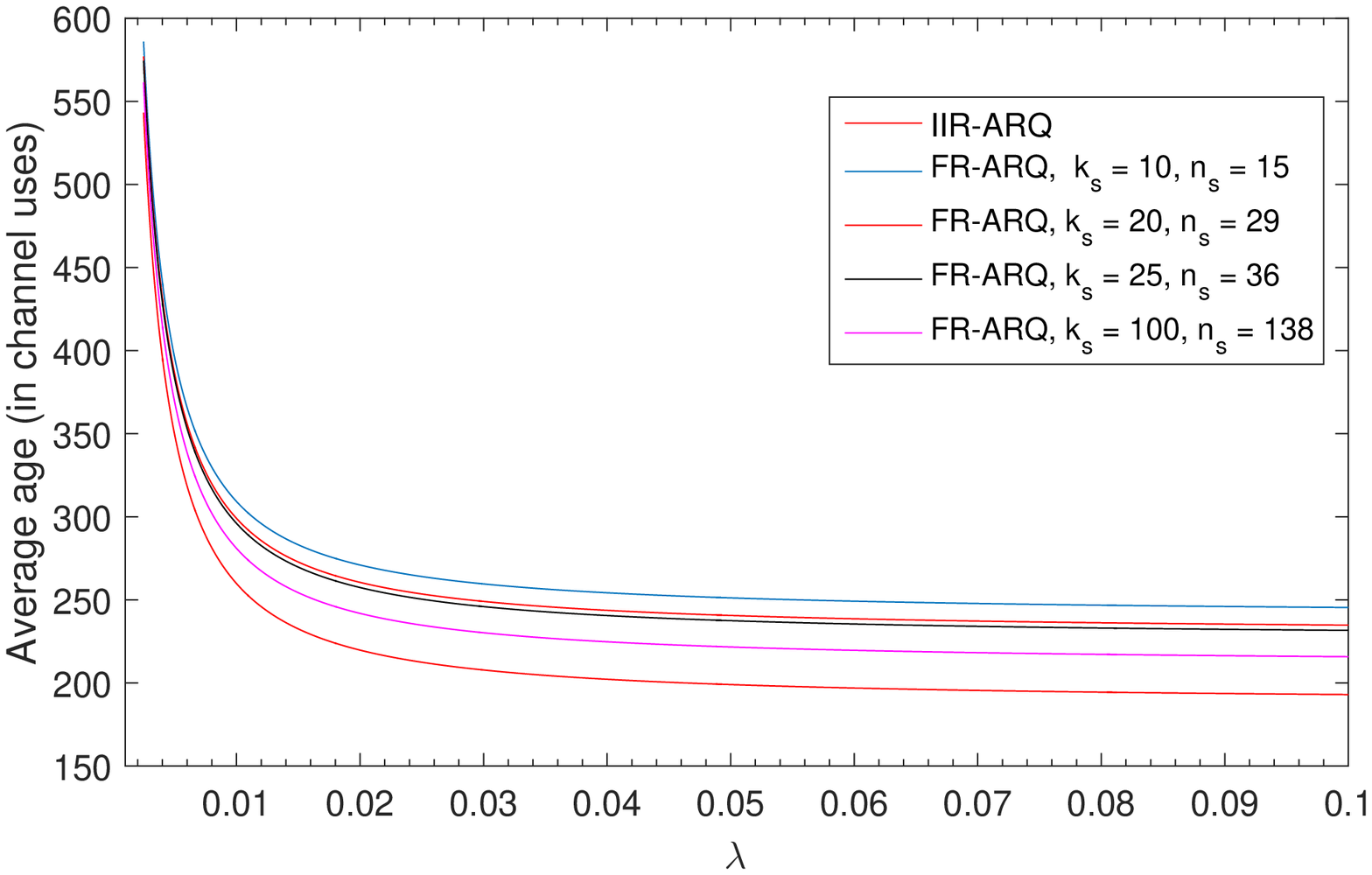}
	\caption{Comparing the performance of the FR-HARQ for the M/G/1/1 without preemption scheme when varying the number of information symbols in each packet. We assume the update has $100$ information symbols, $\delta=0.2$, $k_p=100/k_s$. $n_s$ is chosen to minimize the average age.}
	\label{fig:fig_MG11_no_preempt_diff_ks}
\end{figure}

\begin{figure}[!t]
	\centering
	\includegraphics[scale=0.4]{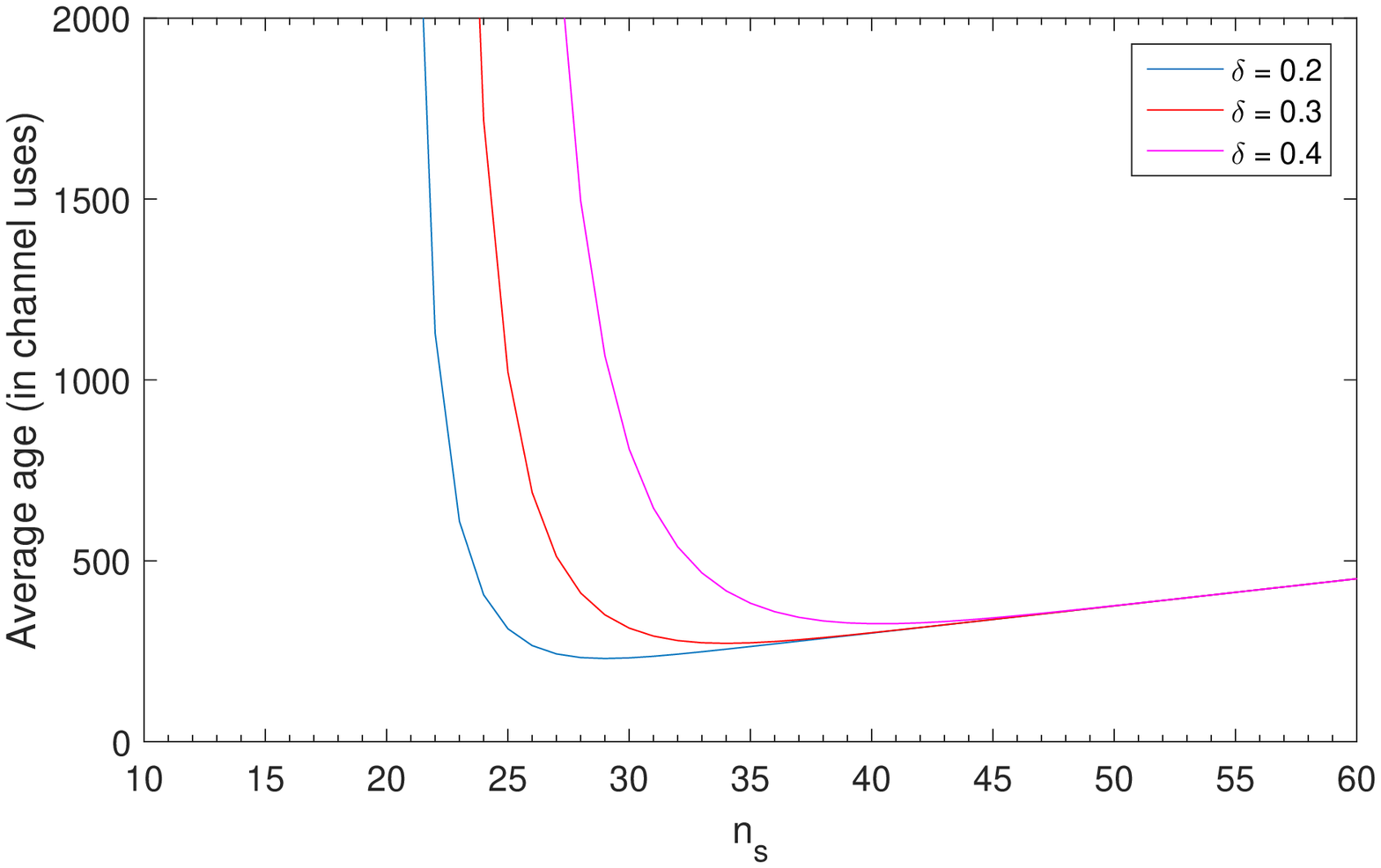}
	\caption{Average age with respect to codeword length for the M/G/1/1 without preemption scheme with FR-HARQ. We assume the update has $100$ information symbols, $\lambda = 1$, $k_s=20$ and $k_p=100/k_s$.}
	\label{fig:fig_MG11_no_preempt_ns}
\end{figure}

\begin{figure}[!t]
	\centering
	\includegraphics[scale=0.4]{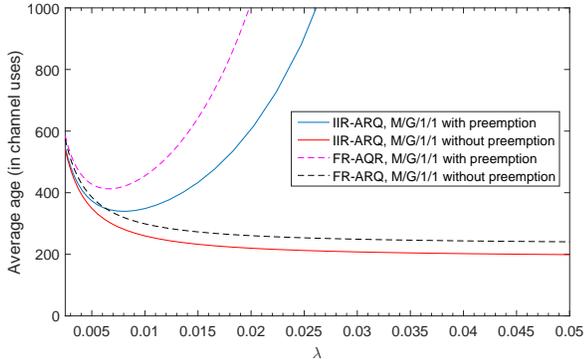}
	\caption{Comparing the performance of the two M/G/1/1 schemes when using IIR and FR. We assume the update has $100$ information symbols and $\delta=0.2$.}
	\label{fig:fig_MG11_all_schemes_IIR_FR}
\end{figure}

%


\section{Conclusion}

In this paper we studied the M/G/1/1 system along with the possible update management policies it presents: preempting the current update or discarding the newly generated one. We derived general expressions for their average age and used this result to compute the average age when considering a practical scenario: updates are sent over a symbol erasure channel using two different HARQ protocols, IIR and FR. In both cases, prioritizing the current update being sent and not preempting it turned out to be the best strategy. Moreover, as it is expected, the IIR protocol gives better performance from an age point of view than FR. Finally, we argued through simulations that for the FR protocol, ensuring reliable delivery of every update packet (by using large codeword length $n_s$) doesn't achieve the optimal average age.


\section*{Acknowledgements}

This research was supported in part by grant No.  200021\_166106/1 of the Swiss National Science Foundation and in part by the NSF Award CIF-1422988.



\bibliographystyle{IEEEtran}
\bibliography{IEEEabrv,ehcache}
%



\end{document}